\documentclass[aps,pra,twocolumn,superscriptaddress,floatfix,nofootinbib,showpacs,longbibliography]{revtex4-1}

\usepackage[utf8]{inputenc}  
\usepackage[T1]{fontenc}     
\usepackage[british]{babel}  
\usepackage[sc,osf]{mathpazo}
\usepackage[table]{xcolor}
\usepackage[scaled=0.86]{berasans}  
\usepackage[colorlinks=true, citecolor=blue, urlcolor=blue]{hyperref}
\makeatletter
\newcommand{\setword}[2]{%
  \phantomsection
  #1\def\@currentlabel{\unexpanded{#1}}\label{#2}%
}
\makeatother
\usepackage{graphicx} 
\usepackage[babel]{microtype}
\usepackage{blkarray}
\usepackage{amsmath}

\usepackage{amsmath,amssymb,amsthm,bm,amsfonts,mathrsfs,bbm} 

\usepackage{xspace}  
\usepackage{pgf,tikz}
\usepackage{xcolor}
\usepackage{multirow}
\usepackage{array}
\usepackage{bigstrut}
\usepackage{braket}
\usepackage{color}
\usepackage{natbib}
\usepackage{multirow}
\usepackage{mathtools}
\usepackage{float}
\usepackage[caption = false]{subfig}
\usepackage{xcolor,colortbl}
\usepackage{color}
\newcommand{\Tr}{\operatorname{Tr}}

\newcommand{\be}{\begin{equation}}
\newcommand{\ee}{\end{equation}}
\newcommand{\ba}{\begin{eqnarray}}
\newcommand{\ea}{\end{eqnarray}}
\newcommand{\ketbra}[2]{|#1\rangle \langle #2|}
\newcommand{\tr}{\operatorname{Tr}}

\newtheorem{theorem}{Theorem}
\newtheorem{corollary}{Corollary}

\def\>{\rangle}
\def\<{\langle}

\usepackage{stmaryrd}

\usepackage{centernot}
\usepackage{subfig}

\usepackage{diagbox}
\usepackage{multirow}
\usepackage{tabularx}

\begin{document}

\title{Advantage of Hardy's Nonlocal Correlation in Reverse Zero-Error Channel Coding}

\author{Mir Alimuddin}
\affiliation{Department of Physics of Complex Systems, S. N. Bose National Center for Basic Sciences, Block JD, Sector III, Salt Lake, Kolkata 700106, India.}
\author{Ananya Chakraborty}
\affiliation{Department of Physics of Complex Systems, S. N. Bose National Center for Basic Sciences, Block JD, Sector III, Salt Lake, Kolkata 700106, India.}

\author{Govind Lal Sidhardh}
\affiliation{Department of Physics of Complex Systems, S. N. Bose National Center for Basic Sciences, Block JD, Sector III, Salt Lake, Kolkata 700106, India.}

\author{Ram Krishna Patra}
\affiliation{Department of Physics of Complex Systems, S. N. Bose National Center for Basic Sciences, Block JD, Sector III, Salt Lake, Kolkata 700106, India.}

\author{Samrat Sen}
\affiliation{Department of Physics of Complex Systems, S. N. Bose National Center for Basic Sciences, Block JD, Sector III, Salt Lake, Kolkata 700106, India.}

\author{Snehasish Roy Chowdhury}
\affiliation{Physics and Applied Mathematics Unit, Indian Statistical Institute, 203 BT Road, Kolkata, India.}

\author{Sahil Gopalkrishna Naik}
\affiliation{Department of Physics of Complex Systems, S. N. Bose National Center for Basic Sciences, Block JD, Sector III, Salt Lake, Kolkata 700106, India.}

\author{Manik Banik}
\affiliation{Department of Physics of Complex Systems, S. N. Bose National Center for Basic Sciences, Block JD, Sector III, Salt Lake, Kolkata 700106, India.}

\begin{abstract}
Hardy's argument constitutes an elegant proof of quantum nonlocality. In this work, we report an exotic application of Hardy's nonlocal correlations in two-party communication setup. We come up with a task, wherein a positive payoff can be through an $1$ bit of communication from the sender to the receiver if and only if the communication channel is assisted with a no-signaling correlation exhibiting Hardy's nonlocality. This further prompts us to establish a counter-intuitive result in correlation assisted reverse zero-error channel coding scenario, where the aim is to simulate a higher input-output noisy classical channel by a lower input-output noiseless one in assistance with preshared correlations. We show that there exist such reverse zero-error channel simulation tasks where non-maximally entangled states are preferable over the assistance with a maximally entangled state, even when the former states carry an arbitrarily small amount of entanglement. Our work thus establishes that within the operational paradigm of local operations and {\it limited} classical communication the structure of entangled resources is even more complex to characterize.   
\end{abstract}


\maketitle	
\section{Introduction} 
The pioneering work of J. S. Bell establishes one of the most striking departures of quantum theory from the profound classical worldview of {\it local-realism} \cite{Bell1966} (see also \cite{Mermin1993,Brunner2014}). Violation of a Bell-type inequality, as demonstrated in several milestone experiments \cite{Clauser1969,Freedman1972,Aspect1981,Aspect1982,Zukowski1993,Weihs1998}, endorses nonlocal nature of the quantum world. Apart from its foundational implications, Bell nonlocality has also been identified as a useful resource for several practical tasks, such as device-independent cryptography \cite{Ekert1991,Barrett2005,Vazirani2014}, reduction of communication complexity \cite{Buhrman2010}, and device-independent randomness certification \& amplification \cite{Pironio2010,Cavalcanti2012,Chaturvedi2015,Colbeck2012}.

A quite popular technique, known by the name 'nonlocality without inequality' proof, is often used to establish the nonlocal behaviour of quantum theory. Unlike the Bell-type inequalities, where statistics of many events are collected, these proofs focus on a single event whose occurrence shows the incompatibility of quantum theory with the notion of local-realism. While the first proof of this kind for tripartite quantum systems is due to Greenberger-Horne-Zeilinger \cite{Greenberger1989} (see also \cite{Greenberger1990}), for bipartite systems, such proof was first proposed by Hardy \cite{Hardy1992}, which is considered to be ``simpler and more compelling than the arguments that underlie the derivation of Bell-CH inequality" \cite{Mermin1994,Clauser1974}. More recently, Hardy-type nonlocality proofs have also been shown to be useful in several practical tasks \cite{Das2013,Mukherjee2015,Li2015,Ramanathan2018,Rai2021,Rai2022}.

In this work, we report a novel application of Hardy's nonlocal correlation in the simplest communication scenario. We show that Hardy's nonlocal correlation shared between two distant parties can empower the communication utility of a perfect classical channel. This is quite striking, as such a correlation by itself cannot be used for information transfer, which otherwise will imply violation of the no signalling (NS) principle. We also argue that the advantage reported in this work is different than the advantage of nonlocal correlations known in  communication complexity tasks \cite{Buhrman2010}.   

We consider a guessing game played between two distant players -- a sender and a receiver. First, we show that the expected collaborative payoff of this game cannot be positive whenever only $1$-bit classical communication is allowed from the sender to the receiver, who otherwise can share an unlimited amount of classical correlation between them. Interestingly, assistance of Hardy's nonlocal correlation to the same classical channel can ensure a strictly positive payoff, establishing a nontrivial advantage of Hardy's correlation in communication tasks. The advantage can be better understood in the framework of correlation-assisted reverse zero-error coding scenario \cite{Cubitt2011}, where the aim is to simulate a higher input-output noisy classical channel by a lower input-output identity channel in the presence of NS correlations. In this scenario, the aforesaid game makes Hardy's nonlocal correlations special as we show that among all the $2$-input-$2$-output NS correlations only those exhibiting Hardy's nonlocality can ensure a positive payoff. This further motivates us to show that pure entangled states which are not maximally entangled are preferable over the maximally entangled one for simulating certain noisy classical channels. At the end we show that similar sort of results can also be obtained by considering a generalization of Hardy's nonlocality argument as proposed by Cabello \cite{Cabello2002}. The present work, therefore, establishes that the comparison of entanglement in quantum states even for bipartite systems are quite complex when classical communication among distant parties are treated as costly.

The manuscript is organized as follows. In Section \ref{game} we introduce a two-party guessing game which we call the Distributed Mine-Hunting game. Here we also discuss the reverse zero-error channel coding setup. In Section \ref{result} we present our main results, and in Section \ref{discussion} we discuss implications of our results along with future outlooks. 

\section{A two-party guessing game}\label{game}
The game involves two distant players, Alice and Bob, and a referee named Charlie. In each run of the game, Charlie provides four closed boxes numbered 1 to 4 to Bob, who has to open one of these boxes. Some of the boxes contain a bomb that will explode upon opening it. Among the boxes that do not contain the bomb, some are empty, some contain a dollar bill, and others may even prompt Bob to pay a dollar bill to Charlie. In each run of the game, Charlie uniformly and randomly picks one of four different arrangements of these boxes, as shown in Fig.\ref{fig1}. Charlie then informs Alice about the arrangement of the boxes in that particular run, and Alice tries to help Bob in choosing a box. However, only $1$-bit of classical communication is allowed from Alice to Bob, which may be further assisted with NS correlations shared between them a priori. From now on, we will refer to this as the Distributed Mine-Hunting (DMH) game.
\begin{figure}[t!]
\centering
\includegraphics[width=0.45\textwidth]{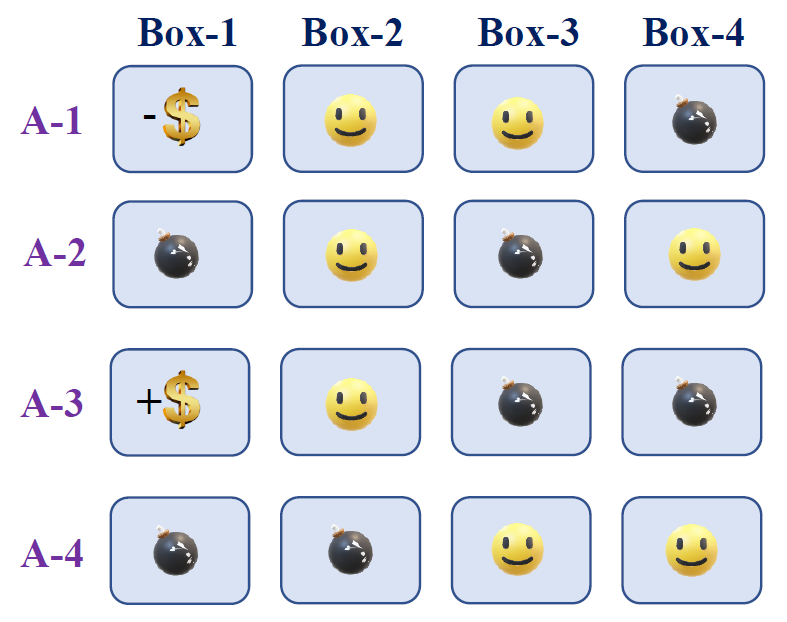}
\caption{Distributed mine-hunting (DMH) game. Opening a box with $`+\$'$ assures dollar bill gain for Bob, whereas a box with $`-\$'$ demands him to pay a dollar bill to Charlie. A box with a `smile' neither offers nor demands any dollar bill, but a box with a `bomb' turns out to be fatal to Bob. Alice knows which of the arrangements $\{$A-$1$,$\cdots$,A-$4$$\}$ Charlie chooses in a particular run and tries with $1$-bit classical channel to help Bob to optimize his expected dollar gain. NS correlation can be used as assistance to the channel.}
\label{fig1}
\end{figure}

\subsection*{General Scenario} 
The aforesaid game can be formally studied within a more generic setup as considered in the reverse zero-error channel coding scenario \cite{Cubitt2011}. In zero-error communication scenario the core goal is to characterize the ability of noisy classical channels to transmit classical information with zero probability of error \cite{Shannon1956} (see also \cite{Korner1998}). A channel from Alice to Bob, with input $m \in \mathcal{M}$ at Alice's end and output $z \in \mathcal{Z}$ at Bob's end, can be represented as a matrix $S \equiv (s_{mz})$, where $s_{mz}$ denotes the probability of producing the output $z \in \mathcal{Z}$ by Bob given that Alice receives the message $m \in \mathcal{M}$. The cardinality of $\mathcal{M}$ and $\mathcal{Z}$ is referred to as the input and output dimensions of the channel, respectively. The direct or forward zero-error coding theorem tries to find the maximum number of distinct input alphabets that can be sent perfectly from Alice to Bob through a noisy channel $S$. In other words, the aim is to find the largest dimensional identity channel ({\it i.e.}, a noise-less channel) that can be simulated by the given channel $S$. The reverse problem, on the other hand, aims to simulate a higher dimensional noisy channel with the help of a lower dimensional identity channel. While both the direct and reverse problems can be studied in the single-shot setup as well as in the asymptotic limit, here we will restrict our study to the reverse problem in the single-shot setup. Interestingly, nonlocal correlations arising from entangled quantum states can provide a nontrivial advantage in the reverse zero-error coding scenario \cite{Cubitt2011}. The present work, however, reports quite an exotic behavior of entanglement by establishing the precedence of non-maximally entangled states over the maximal one in some instances of  reverse zero-error channel coding problem.     
  
Within the aforesaid notation, a game (such as DMH) is entirely specified by the payoff matrix $\mathcal{G}\equiv(g_{mz})$, where $g_{mz}\in\overline{\mathbb{R}}$ is the reward/payoff given when Bob produced the index `$z$' provided Alice received the message `$m$'. For instance, the DMH game is specified by the following payoff matrix:
\begin{align}
\mathcal{G}_{DMH}	
\equiv
\begingroup
\setlength{\tabcolsep}{1pt} 
\renewcommand{\arraystretch}{1} 
\begin{array}{c||c|c|c|c|} 
 \mathcal{M}\textbackslash\mathcal{Z}& 1 & 2 & 3 & 4\\ \hline\hline
1& -1 & 0 & 0 & -\infty\\ \hline
2& -\infty & 0 & -\infty & 0\\ \hline
3& +1 & 0 & -\infty & -\infty\\ \hline
4& -\infty & -\infty & 0 & 0\\ \hline
\end{array}
\endgroup
\label{DMH-payoff}
\end{align}
Payoffs in (\ref{DMH-payoff}) quantitatively capture the scenario of the DMH game. A reward of $-\infty$ for the box containing the bomb captures the notion that choosing such a box must be avoided at all costs \cite{Self0}. The reward $0$ corresponds to the event where the players survive but do not receive any reward. Events with reward $+1~(-1)$ correspond to the scenario where the players receive (pay) some dollar bill from (to) Charlie. The game matrix and the sampling distribution of Alice's inputs are common knowledge to the players.

Alice and Bob are cooperative in nature and aim to maximize the payoff. Their collaborative strategy depends on the available resources, which can be broadly categorized into two types: (i) correlation shared between them before the game starts and (ii)  communication from Alice to Bob. Any strategy employed by the players can be represented as an input-output channel $S\equiv(s_{mz})$. Given such a strategy matrix $S$, the average payoff can be obtained as
\begin{align}\label{pay}
\mathcal{P}(S)=\sum_{z,m} p(m)~g_{mz}~s_{mz}~.
\end{align}
As it is evident, there will always be a perfect strategy for such a game if $\log_2|\mathcal{M}|$ bits of communication are allowed from Alice to Bob. Interesting situations arise when communication is limited, which can further be aided by preshared correlations of different kinds. 

\section{Results}\label{result} 
Let $\Omega_{n_c+SR}(|\mathcal{M}|,|\mathcal{Z}|)$ denote the set of strategy matrices obtained when $n$-bits of classical communication and an unlimited amount of shared randomness are available. The set $\Omega_{n_c+SR}(|\mathcal{M}|,|\mathcal{Z}|)$ forms a polytope with extreme points $S^{e}$'s representing strategy matrices obtained through deterministic encoding $\mathbb{E}:\mathcal{M}\to\{0,1\}^n$ at Alice's end and deterministic decoding $\mathbb{D}:\{0,1\}^n\to\mathcal{Z}$ at Bob's end \cite{Frenkel2015} (see also \cite{Dallarno2017}). Our first technical result is to limit the optimal success probability of the game in Eq.(\ref{DMH-payoff}) for classical strategies with limited communication.
\begin{figure}[t!]
\centering
\includegraphics[width=0.45\textwidth]{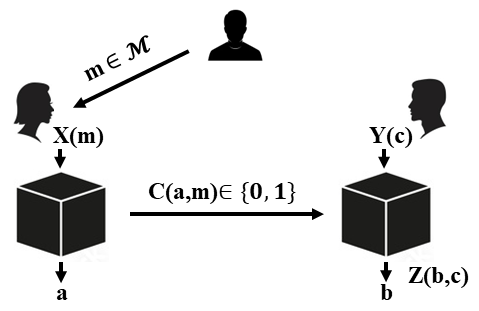}
\caption{General strategy to play a game $\mathcal{G}$ when the $1$-cbit communication channel is assisted with NS correlation. Alice computes the input $x=X(m)$ to her part of the nonlocal box based on the message $m\in\mathcal{M}$ received from Referee. The output $a$ of the nonlocal box at her end and the message $m$ determine the classical bit $c=C(a,m)$ sent to Bob. Bob then inputs $y=Y(c)$ into his end of the NS box obtaining the output $b$. Finally, he generates his guess as $z=Z(b,c)$. }
\label{fig2}
\end{figure}
 \begin{theorem}\label{theo1} 
The average payoff of the DMH game is upper bounded by zero while following a strategy from the set $\Omega_{1_c+SR}(4,4)$, {\it i.e.}, $\mathcal{P}(S)\le0,~\forall~S\in\Omega_{1_c+SR}(4,4)$.
 \end{theorem}
 \begin{proof}
Since average payoff depends linearly on the strategy [see Eq.(\ref{pay})], the polytope structure of $\Omega_{1_c+SR}(4,4)$ ensures that maximum payoff will be attained at one of its vertices. For $\Omega_{n_c+SR}(|\mathcal{M}|,|\mathcal{Z}|)$, the number of vertices $N$ can be calculated using the formula \cite{Dallarno2017}: 
\begin{align}
N=\sum^{2^n}_{k=1}&k!
\left( {\begin{array}{c}
    |\mathcal{Z}| \\
   k 
  \end{array} } \right) \left\{ {\begin{array}{c}
    |\mathcal{M}| \\
   k \\
  \end{array} } \right\}, \\  
  \left( {\begin{array}{c}
    |\mathcal{Z}| \\
   k 
  \end{array} } \right)&:=\frac{|\mathcal{Z}|!}{k!(|\mathcal{Z}|-k)!},\nonumber\\
  \left\{ {\begin{array}{c}
    |\mathcal{M}| \\
   k \\
  \end{array} } \right\}&:=\sum^{k}_{j=0}\frac{1}{k!}(-1)^{k-j}  \left( {\begin{array}{c}
k \\
   j 
  \end{array} } \right) j^{|\mathcal{M}|}.\nonumber
\end{align}
Note that the players' first priority is to avoid the bomb at any cost, {\it i.e.}, the $-\infty$ reward in Eq.(\ref{DMH-payoff}). This will exclude some vertices. For instance, consider the extreme strategy matrix,     
\begin{align}\label{cm1}
S^e_{ns}:=\begin{pmatrix}
     1 & 0 & 0 & 0\\
     0 & 0 & 1 & 0\\
     1 & 0 & 0 & 0\\
     0 & 0 & 1 & 0
\end{pmatrix}.
\end{align}
Comparing with Eq.(\ref{DMH-payoff}), it is evident that in this strategy the bomb will be triggered with a nonzero probability (for $m=2$). Among all possible vertices, only in the following five cases 
\begin{align}
\left\{\begin{aligned}
S^e_{s1}:=\begin{pmatrix}
     1 & 0 & 0 & 0\\
     0 & 0 & 0 & 1\\
     1 & 0 & 0 & 0\\
     0 & 0 & 0 & 1
     \end{pmatrix},~~
S^e_{s2}:=\begin{pmatrix}
     0 & 1 & 0 & 0\\
     0 & 1 & 0 & 0\\
     0 & 1 & 0 & 0\\
     0 & 0 & 1 & 0
     \end{pmatrix},\\
S^e_{s3}:=\begin{pmatrix}
     0 & 1 & 0 & 0\\
     0 & 1 & 0 & 0\\
     0 & 1 & 0 & 0\\
     0 & 0 & 0 & 1
     \end{pmatrix},~~
S^e_{s4}:=\begin{pmatrix}
     0 & 1 & 0 & 0\\
     0 & 0 & 0 & 1\\
     0 & 1 & 0 & 0\\
     0 & 0 & 0 & 1
     \end{pmatrix},\\
S^e_{s5}:=\begin{pmatrix}
     0 & 0 & 1 & 0\\
     0 & 1 & 0 & 0\\
     0 & 1 & 0 & 0\\
     0 & 0 & 1 & 0
     \end{pmatrix}~~~~~~~~~~~~~~~~
\end{aligned}\right\}
\end{align}
the bomb will never be triggered. However, we have $\mathcal{P}(S^e_{sK})=0,~\forall~K\in\{1,\cdots,5\}$. Since all the deterministic strategies result in either zero or $-\infty$ payoff,  the optimal payoff with $1$ bit classical communication and shared randomness is simply zero. This completes the claim. 
\end{proof}
We now consider the scenario where the communication line from Alice to Bob is aided with a generic NS correlation $\mathbb{P}\equiv\{p(a,b|x,y)\}$, where $p(a,b|x,y)\ge0,~\forall~a,b,x,y~\&~\sum_{a,b}p(a,b|x,y)=1,~\forall~x,y$. Here $p(a,b|x,y)$ denotes the probability of obtaining outcome $a\in\mathcal{A}$ at Alice's end and $b\in\mathcal{B}$ at Bob's end for their respective inputs $x\in\mathcal{X}$ and $y\in\mathcal{Y}$. Classical correlations that allow a local-realistic description, $p(a,b|x,y)=\int_{\Lambda}\mu(\lambda)p(a|x,\lambda)p(b|y,\lambda)d\lambda$, forms a strict subset (a sub polytope) of the NS polytope; here $\lambda\in\Lambda$ is some classical variable shared between Alice and Bob and $\mu(\lambda)$ is a probability distribution on $\Lambda$ \cite{Brunner2014}. Quite surprisingly, entangled quantum states can lead to correlations that are not local-realistic and nonlocality of those correlations can be certified through violation of some Bell-type inequalities \cite{Bell1966}. Given such an NS correlation (possibly nonlocal) as an assistance to $1$-bit classical communication from Alice to Bob, the general strategy to play a game $\mathcal{G}$ is described in Fig.\ref{fig2}.  

For the binary input-output case, Lucien Hardy proposed an elegant argument according to which any NS correlation $\mathbb{H}\equiv\{h(a,b|x,y)\}$ satisfying the constraints
\begin{align}
\left\{\begin{aligned}
h(0,0|0,0):=h_0>0,~~h(0,1|0,1):=h_5=0,\\
h(1,0|1,0):=h_{10}=0,~~h(0,0|1,1):=h_3=0,\\
\mbox{with},~h(a,b|x,y):=h_{a\times2^3+b\times2^2+x\times2^1+y\times2^0}
\end{aligned}\right\}\label{Hardy}
\end{align}
must be nonlocal in nature \cite{Hardy1992} (for the sake of completeness we analyze the argument in Appendix-\ref{appen-a}). A correlation exhibiting Hardy's nonlocality is termed as Hardy's nonlocal correlation. Our next result proves a nontrivial advantage of Hardy correlation while playing the DMH game. 
\begin{theorem}\label{theo2}
A strictly positive average payoff in the DMH game is achievable with $1$-bit perfect classical channel from Alice to Bob when the channel is assisted with a $2$-input-$2$-output Hardy's nonlocal correlation. 
\end{theorem}
\begin{proof}
Consider the following strategy by Alice and Bob.
{\bf Alice's action:}\\
$\bullet$ Depending on $m\in\mathcal{M}$ Alice computes her input in the NS box. For $m\in\{1,3\}$ she chooses $x=0$, otherwise she choose $x=1$.\\
$\bullet$ Based on the tuple $(m,a)\in\mathcal{M}\times\mathcal{A}$ she communicates to Bob. She sends $c=0$ to Bob when $(m,a)\in\{(1,1),(2,1),(3,0),(3,1)\}$, else she sends $c=1$.\\
{\bf Bob's action:}\\
$\bullet$ The communicated bit from Alice is used as input in Bob's part of the NS box.\\
$\bullet$ Depending on the tuple $(c,b)\in C\times\mathcal{B}$ he chooses the box as follows: $(0,0)\mapsto1,~(0,1)\mapsto2,~(1,0)\mapsto3,~(1,1)\mapsto4$.

The above strategy with $1$-bit communication and $2$-input-$2$-output Hardy's correlation $\mathbb{H}\equiv\{h(a,b|x,y\}$ leads to the strategy matrix,
\begin{align}
S_{\mathbb{H}}	
\equiv
\begingroup
\setlength{\tabcolsep}{1pt} 
\renewcommand{\arraystretch}{1} 
\begin{array}{c||c|c|c|c|} 
 \mathcal{M}\textbackslash\mathcal{Z}& 1 & 2 & 3 & 4\\ \hline\hline
1& h_8 & h_{12} & h_1 & 0\\ \hline
2& 0 & h_{14} & 0 & h_7\\ \hline
3& h_0+h_8 & h_4+h_{12} & 0 & 0\\ \hline
4& 0 & 0 & h_{11} & h_7+h_{15}\\ \hline
\end{array}
\endgroup
\label{HardyChannel}
\end{align}
As evident from the payoff matrix (\ref{DMH-payoff}) and the strategy (\ref{HardyChannel}), a box containing bomb will never be opened. Furthermore, assuming Charlie's choice to be completely random, we have the average payoff 
\begin{align}\label{pay}
\mathcal{P}(S_{\mathbb{H}})=\frac{1}{4}\tr[ \mathcal{G}^{T}_{DMH}~S_{\mathbb{H}}]=\frac{1}{4}h_0>0.
\end{align}
This completes the proof. 
\end{proof}
This in turn establishes that the $4$-input $4$-output noisy channel $S_{\mathbb{H}}$ can be perfectly simulated by the $2$-dimensional identity channel ({\it i.e.} a $1$-bit perfect channel from Alice to Bob) when assisted with Hardy's nonlocal correlation. However, as follows form Theorem \ref{theo1}, assistance of arbitrary amount of shared randomness fails to achieve the goal. A correlation with Hardy success $h_0$ yields Clauser-Horne-Shimony-Holt (CHSH) \cite{Clauser1969(1)} value $2+4~h_0$ \cite{Cereceda2000}. Accordingly, the CHSH value corresponding to the optimal quantum Hardy correlation is strictly less than the Cirel'son bound \cite{Rabelo2012, Cirelson1980}. Therefore, a natural question is whether other nonlocal quantum correlations can lead to better success in the DMH game. Our next result answers this question in negation.     
\begin{theorem}\label{theo3}
Any $2$-input-$2$-output NS correlation providing a strictly positive payoff in the DMH game as an assistance to the $1$-bit of perfect classical channel must exhibit Hardy's nonlocality.
\end{theorem}
\begin{proof}
The set of strategy matrices simulable by $1$-cbit communication with the assistance of a given NS correlation $\mathbb{P}\equiv\{p(a,b|x,y)~|~a,b,x,y\in\{0,1\}\}$ and unlimited SR forms a polytope. We denote this set by $\Omega_{1_c+SR+P}$. Vertices of this polytope correspond to strategies where the players follow an encoding and decoding scheme characterized by deterministic functions of the form $x=X(m),~ c=C(m,a),~ y=Y(c),~ z=Z(c,b)$. The linearity of the payoff function again implies that the maximum payoff occurs at one of the vertices of $\Omega_{1_c+SR+P}$.

Elementary counting shows that there are $2^4\times 2^8 \times 2^2 \times 4^4$ such deterministic strategies. Furthermore, elements of the strategy matrix are related linearly to the NS correlation $P=\{p(a,b|x,y)\}$,
\begin{align}\label{lin}
s_{mz}&=\sum_{x,a,c,b,y\in\{0,1\}} \delta_{x,X(m)}\times\delta_{c,C(m,a)}\times\delta_{y,Y(c)}\nonumber\\
&\hspace{3.5cm}\times\delta_{z,Z(c,b)}\times p(a,b|x,y).
\end{align}
For a non-negative payoff, the game matrix enforces some of the entries of $S\equiv(s_{mz})$ to be zero. Since Eq.(\ref{lin}) is linear, one can solve these equality constraints to see what restrictions are imposed on the NS correlation $P(a,b|x,y)$. By brute-forcing through all the deterministic strategies and by symbolic programming, we were able to verify that for each vertex of $\Omega_{1_c+SR+P}$, the positivity of average payoff imposes the conditions in Eq.(3) in the main manuscript (or its local reversible relabelling) to the NS correlation $\{p(a,b|x,y)\}$. This proves the claim that among all $2$-input-$2$-output correlations, only Hardy's nonlocal correlations can provide a positive payoff in the DMH game.
\end{proof}
Although it is known that a two-qubit maximally entangled state does not exhibit Hardy's nonlocality \cite{Goldstein1994,Jordan1994}, still Theorem \ref{theo3} is not sufficient to make a claim that such a state shared between Alice and Bob cannot lead to a strategy yielding strictly positive payoff in DMH game. Increasing the cardinality of the input-output sets $\mathcal{X},~\mathcal{Y},~\mathcal{A},~\mathcal{B}$ Alice and Bob can generate a more general NS correlation and then try to utilize this correlation to assist the $1$-cbit channel to obtain a nonzero payoff in DMH game. However, our next result proves a no-go to this aim. 
\begin{theorem}\label{theo4}
Two qubit maximally entangled state together with $1$-bit perfect classical channel from Alice to Bob does not result in a strategy ensuring a strictly positive average payoff in the DMH game.  
\end{theorem}
\begin{proof}
Let Alice and Bob share the two-qubit maximally entangled state $\ket{\phi^+}=\frac{1}{\sqrt{2}}\left(\ket{00}+\ket{11}\right)$.
Note that the existence of a strategy involving SR that yields an advantage in winning the game necessarily implies the existence of a strategy without involving any SR. We now prove that such a strategy does not exist. Without loss of generality, we can assume that Alice does a $2$ outcome measurement depending on the classical message $m$ she receives, {\it i.e.}, she performs $\{E^{(m)}_0,E^{(m)}_1~|~E^{(m)}_0+E^{(m)}_1=\mathbb{I}\}$ for $m\in\{1,\cdots,4\}$. Alice then communicates $c\in\{0,1\}$ if the outcome $E^{(m)}_c$ clicks. Based on the communicated bit $c$, Bob performs a $4$ outcome measurement, and based on the measurement outcome, he chooses a box. Bob's measurement is denoted by $\{N^{(c)}_z\}_{z=1}^4$ with $\sum_{z=1}^4 N^{(c)}_z=\mathbb{I}$ when the communication $c\in\{0,1\}$ is received. This leads to the strategy matrix with elements
\begin{align}
s_{mz}&=\sum_{c=0}^1 s_{mzc}:=\sum_{c=0}^1 \Tr\left[\ketbra{\phi^+}{\phi^+}\left(E^{(m)}_c\otimes N^{(c)}_z\right)\right],\nonumber\\
&=\frac{1}{2}\sum_{c=0}^1 \Tr\left[E^{(m)}_c N^{*(c)}_z\right]. \label{phistrat}
\end{align}
The aim is to find POVM elements $E^{(m)}_c$ and $N^{*(c)}_z$ such that the resulting strategy yields a positive payoff in the DMH game. For simplicity, we drop the notation `$*$' in $N^{*(c)}_z$ keeping in mind that if we do indeed find a solution for Eq.(\ref{phistrat}), we need to complex conjugate the matrices $N^{(c)}_z$. We also note that if $N^{(c)}_z$ forms a measurement so does $N^{*(c)}_z$. Therefore a strictly positive 
payoff in the DMH game demands the following conditions to be satisfied: 
\begin{subequations}
\begin{align}
 \Tr\left[E^{(1)}_0 N^{(0)}_4\right]=0=\Tr\left[E^{(1)}_1 N^{(1)}_4\right],\label{a}\\ 
\Tr\left[E^{(2)}_0 N^{(0)}_1\right]=0=\Tr\left[E^{(2)}_1 N^{(1)}_1\right],\label{b}\\
\Tr\left[E^{(2)}_0 N^{(0)}_3\right]=0=\Tr\left[E^{(2)}_1 N^{(1)}_3\right],\label{c}\\
\Tr\left[E^{(3)}_0 N^{(0)}_3\right]=0=\Tr\left[E^{(3)}_1 N^{(1)}_3\right],\label{d}\\
\Tr\left[E^{(3)}_0 N^{(0)}_4\right]=0=\Tr\left[E^{(3)}_1 N^{(1)}_4\right],\label{e}\\
\Tr\left[E^{(4)}_0 N^{(0)}_1\right]=0=\Tr\left[E^{(4)}_1 N^{(1)}_1\right],\label{f}\\
\Tr\left[E^{(4)}_0 N^{(0)}_2\right]=0=\Tr\left[E^{(4)}_1 N^{(1)}_2\right],\label{g}
\end{align}
\label{eq}
\vspace{-.8cm}
\end{subequations}
\begin{align}
\sum_{c=0}^1 \Tr\left[E^{(3)}_c N^{(c)}_1\right]>\sum_{c=0}^1 \Tr\left[E^{(1)}_c N^{(c)}_1\right].\label{ineq}
\end{align}
The inequality (\ref{ineq}) can further be written as 
\begin{align}
 \Tr\left[E^{(3)}_0\left(N^{(0)}_1-N^{(1)}_1\right)\right]>\Tr\left[E^{(1)}_0\left(N^{(0)}_1-N^{(1)}_1\right)\right].\label{ineq2}
\end{align}
Note that, if a solution of POVMs exists satisfying (\ref{eq}) and (\ref{ineq2}), then there also exists a solution with  $\{E^{(m)}_0, E^{(m)}_1\}$ being projective measurements $\forall~m\in\{1,\cdots,4\}$. This can be argued easily by expanding all the operators $\{E^{(m)}_0,E^{(m)}_1\}$ in the spectral form. For example in inequality (\ref{ineq2}) on the LHS, we can choose the projector formed by the eigenvector of $E^{(3)}_0$ giving the maximum value of the trace. Similarly, on the RHS, we can choose the projector formed by the eigenvector of $E^{(1)}_0$ giving the lowest value of trace. Moreover, notice than for Eqs.(\ref{eq}) all the projectors corresponding to different eigenvalues of  $\{E^{(m)}_0,E^{(m)}_1\}$ must satisfy the Eqs.(\ref{eq}) individually. Thus we start by assuming all measurements$\{E^{(m)}_0,E^{(m)}_1\}$ are projective. In particular, let
\begin{align*}
E^{(3)}_0=\ketbra{\psi}{\psi},~~~~E^{(3)}_1=\ketbra{\psi^{\perp}}{\psi^{\perp}},\\
E^{(4)}_0=\ketbra{\phi}{\phi},~~~~E^{(4)}_1=\ketbra{\phi^{\perp}}{\phi^{\perp}}.
\end{align*}
Thus looking into LHS of Eqs.(\ref{d}) \& (\ref{g}) we have
\begin{align*}
N^{(0)}_1=p_1\ketbra{\phi^{\perp}}{\phi^{\perp}},~~~~N^{(0)}_2=p_2\ketbra{\phi^{\perp}}{\phi^{\perp}},\\
N^{(0)}_3=p_3\ketbra{\psi^{\perp}}{\psi^{\perp}},~~~~N^{(0)}_4=p_4\ketbra{\psi^{\perp}}{\psi^{\perp}},   
\end{align*}
with $0\leq p_1,\cdots,p_4\leq 1$. Since $\{N^{(0)}_1,\cdots,N^{(0)}_4\}$ form a measurement, we therefore have 
\begin{align*}
\left(p_1+p_2\right)\ketbra{\phi^{\perp}}{\phi^{\perp}}+\left(p_3+p_4\right)\ketbra{\psi^{\perp}}{\psi^{\perp}}=\mathbb{I},
\end{align*}
which will be satisfied {\it if and only if} 
\begin{align*}
p_1+p_2=p_3+p_4=1,~~\&~~\ket{\phi}=\ket{\psi^{\perp}}. 
\end{align*}
A similar argument can be made for the other measurement of Bob and we get the following
\begin{subequations}
\begin{align}
N^{(0)}_1&=p_1\ketbra{\psi}{\psi},~~~~~~~~~N^{(1)}_1=q_1\ketbra{\psi^{\perp}}{\psi^{\perp}},\label{aa}\\
N^{(0)}_2&=p_2\ketbra{\psi}{\psi},~~~~~~~~~N^{(1)}_2=q_2\ketbra{\psi^{\perp}}{\psi^{\perp}},\label{bb}\\
N^{(0)}_3&=p_3\ketbra{\psi^{\perp}}{\psi^{\perp}},~~~~N^{(1)}_3=q_3\ketbra{\psi}{\psi},\label{cc}\\
N^{(0)}_4&=p_4\ketbra{\psi^{\perp}}{\psi^{\perp}},~~~~N^{(1)}_4=q_4\ketbra{\psi}{\psi},\label{dd}
\end{align}
\end{subequations}
where $0\leq q_1,\cdots,q_4\leq 1$ and $q_1+q_2=q_3+q_4=1$. Now let $E^{(1)}_0=\ketbra{\chi}{\chi}$ and $E^{(1)}_1=\ketbra{\chi^{\perp}}{\chi^{\perp}}$. From LHS of Eq.(\ref{a}) and LHS of Eq.(\ref{dd}) we must have $p_4\ketbra{\psi^{\perp}}{\psi^{\perp}}=s_1\ketbra{\chi^{\perp}}{\chi^{\perp}}$, with $0\leq s_1\leq 1$. A solution is $p_4=s_1>0$ and $\ket{\psi}=\ket{\chi}$ which implies  $E^{(1)}_0=E^{(3)}_0$ and $E^{(1)}_1=E^{(3)}_1$. This further tells us that Alice uses the same strategy for $m=1$ and $m=3$ and this will never yield a positive payoff.  The only other solution is  
\begin{align*}
p_4=s_1=0\implies N^{(0)}_4=0 \implies N^{(0)}_3=\ketbra{\psi^{\perp}}{\psi^{\perp}}.
\end{align*}
Similarly, we can argue that we must have $N^{(1)}_3=\ketbra{\psi}{\psi}$. Now from Eq.(\ref{c}) we have $E^{(2)}_0\propto \ketbra{\psi}{\psi}$ and $E^{(2)}_1\propto \ketbra{\psi^{\perp}}{\psi^{\perp}} \implies E^{(2)}_0=\ketbra{\psi}{\psi}$ and $E^{(2)}_1=\ketbra{\psi^{\perp}}{\psi^{\perp}}$. From Eq.(\ref{b}) and Eq.(\ref{aa}), the only solution is $N^{(0)}_1=N^{(1)}_1=0$, which leads to violation of the inequality in Eq.(\ref{ineq}). Thus a consistent solution cannot be found satisfying the conditions (\ref{eq}) \&~(\ref{ineq}). This completes the proof.
\end{proof}
This theorem has an interesting implication. It shows that there exists a communication task wherein a non-maximally pure entangled state can be preferable over the maximally entangled one even when the entanglement of the former is vanishingly zero. More formally we can deduce the following corollary.
\begin{corollary}\label{coro1}
For every non maximally entangled state $\ket{\psi}\in\mathbb{C}^2\otimes\mathbb{C}^2$ there exists a strategy matrix $S_{\psi}$ such that $S_{\psi}\in\Omega_{1_c+SR+\ket{\psi}}(4,4)$ but $S_{\psi}\notin\Omega_{1_c+SR+\ket{\phi^+}}(4,4)$. 
\end{corollary}
Here $\ket{\phi^+}:=(\ket{00}+\ket{11})/\sqrt{2}$ and $\Omega_{1_c+SR+\ket{\chi}}(4,4)$ denotes the convex set of strategy matrices simulable with $1$-cbit communication from Alice to Bob when the communication channel is further assisted with preshared quantum state $\ket{\chi}$ and unlimited shared randomness. The Corollary \ref{coro1} follows when results of Theorems \ref{theo2} \& \ref{theo4} are combined with the fact that all non maximally pure entangled state exhibits Hardy's nonlocality\cite{Goldstein1994}. More precisely, for every non maximally entangled state $\ket{\psi}\in\mathbb{C}^2\otimes\mathbb{C}^2$ there exists a noisy channel of the form $S_{\mathbb{H}}$, that can be perfectly simulated with $1$-bit perfect classical channel when assisted with the state $\ket{\psi}$, but not with $\ket{\phi^+}$. A more detailed discussion on the implications of theses results is presented in Appendix-\ref{appen-b}.

\section{Discussions and outlook} \label{discussion}
The seminal quantum superdense coding protocol is worth mentioning as it shows that quantum entanglement, pre-shared between a sender and a receiver, can increase the classical communication capacity of a quantum system \cite{Bennett1992} (see also \cite{Bennett1999}). While in quantum superdense coding protocol a quantum channel is considered, here we show that quantum entanglement can even empower the communication utility of a perfect classical channel. As already mentioned, such an advantage can be better understood in zero-error and reverse zero-error communication set-up.  While such an advantage of quantum entanglement is already known (see Proposition $21$ in \cite{Cubitt2011} and see also \cite{Cubitt2010,Frenkel2022,Patra2022}), the full picture of entanglement assistance is not well understood. Our Theorem \ref{theo4} proves a nontrivial result to this direction. It shows that in the correlation-assisted reverse zero-error coding scenario, there exists noisy channel simulation tasks wherein non-maximally entangled states, with arbitrarily less amount of entanglement, are preferable over maximally entangled state. In the resource theory of quantum entanglement \cite{Plenio2007}, where local operations and classical communication (LOCC) are considered to be free, a maximally entangled state is more useful than a non-maximally entangled one, and a deterministic LOCC transformation is always possible from the former to the latter \cite{Nielsen1999}. Our work, however, establishes that within the operational paradigm of local operations and {\it limited} classical communication the structure of entangled resources are quite complex to characterize.
 
A nonlocal advantage is also known in a variant of the communication scenario known as the communication complexity problem \cite{Buhrman2010}. In such a scenario, Bob's goal is not to determine Alice's data $\mathcal{M}$ but to determine some information that is a function of $\mathcal{M}$ in a way that may depend on the other data $\mathcal{N}$ that resides with Bob while $\mathcal{N}$ is unknown to Alice. In that sense, our scenario is 
closer to the standard framework of Shannon \cite{Shannon1948} (and considered by Holevo in quantum set up \cite{Holevo1973}), where at Bob's end, no further data set $\mathcal{N}$ is considered, albeit in single-shot setup.

In conclusion, the present work establishes exotic uses of quantum entanglement in zero-error information theory \cite{Shannon1956} (see also \cite{Korner1998}) whose motivation arises from the fact that in many real-world critical applications, no errors can be tolerated, and in practice, the communication channel can only be available for a finite number of times. In particular, we show that quantum correlations exhibiting Hardy's nonlocality can empower the communication {\it utility} of a perfect classical communication channel. In Appendix \ref{appen-c} we show that similar results can be obtained by considering the generalization of Hardy's nonlocality argument as proposed by Cabello \cite{Cabello2002}. Our work also motivates many questions for future study. For instance, it would be interesting to see whether any nonlocal correlation can be made useful as a communication resource in the sense discussed here. It will also be interesting to see whether maximally entangled states of higher dimensions provide some advantage in the DMH game. More generally, characterizing the set $\Omega_{n_c+SR+\chi}(|\mathcal{M}|,|\mathcal{Z}|)$ for an arbitrary quantum state $\ket{\chi}$ would be very interesting. 

\begin{acknowledgements}
We thankfully acknowledge discussions with Ashutosh Rai. GLS acknowledges support from the CSIR project 09/0575(15830)/2022-EMR-I. SGN acknowledges support from the CSIR project 09/0575(15951)/2022-EMR-I. SRC acknowledges support from University Grants Commission. MA and MB acknowledge funding from the National Mission in Interdisciplinary Cyber-Physical systems from the Department of Science and Technology through the I-HUB Quantum Technology Foundation (Grant no: I-HUB/PDF/2021-22/008). MB acknowledges support through the research grant of INSPIRE Faculty fellowship from the Department of Science and Technology, Government of India, and the start-up research grant from SERB, Department of Science and Technology (Grant no: SRG/2021/000267).
\end{acknowledgements}

\appendix
\section{Hardy's nonlocality argument}\label{appen-a}
Hardy's argument is a popular method to check Bell nonlocality of a given NS correlation \cite{Hardy1992}. A $2$-input-$2$-output NS correlation $\mathbb{H}\equiv\{h(a,b|x,y)\}$, with $a,b,x,y\in\{0,1\}$ will exhibit Bell nonlocality if they satisfy the constraints 
\begin{subequations}
\begin{align}
h(0,0|0,0)&>0, \label{ha}\\
h(0,1|0,1)&=0, \label{hb}\\
h(1,0|1,0)&=0,\label{hc}\\
h(0,0|1,1)&=0. \label{hd}
\end{align}
\end{subequations}
Recall that `Bell locality` condition demands that such a correlation $\{h(a,b|x,y)\}$ must be factorizable in the form
\begin{align}\label{fact}
h(a,b|x,y)&=\int_{\lambda \in \Lambda} d\lambda \mu(\lambda)p_A(a|x,\lambda)p_B(b|y,\lambda),\\ 
&~~~~~~~~~~~~~~~~~~~~~~\forall~a,b,x,y,\nonumber
\end{align}
where $\lambda \in \Lambda$ is the local 
hidden variable, $\mu(\lambda)$ is the distribution of the hidden variable over $\Lambda$, which according to `freedom of choice' is assumed to be independent of the Alice's and Bob's choices of inputs, {\it i.e.}, $\mu(\lambda|x,y)=\mu(\lambda)$; and $p_{_X}(i|j,\lambda)$ is the probability that party $X$ observes the outcome $i$ given that they have performed the measurement $j$ when the hidden state is $\lambda$. Applying Eq.(\ref{fact}) to Eq.(\ref{ha}) we can say that there exist a set $\Tilde{\Lambda} \subseteq \Lambda$ of nonzero measure w.r.t. $\Lambda$ which can be defined as 
\begin{align*}
\Tilde{\Lambda} \equiv \{\lambda\in\lambda~|~\mu(\lambda)>0,p_A(0|0,\lambda)>0,p_B(0|0,\lambda)>0\}.  
\end{align*}
Now, Eq(\ref{fact}) can be written as
\begin{align}
h(a,b|x,y)=&\int_{\lambda \in \Tilde{\Lambda}} d\lambda \mu(\lambda)p_A(a|x,\lambda)p_B(b|y,\lambda) \nonumber \\
+&\int_{\lambda \in \Tilde{\Lambda}^c} d\lambda \mu(\lambda)p_A(a|x,\lambda)p_B(b|y,\lambda),\label{factc}
\end{align}
where $\Tilde{\Lambda}^c$ is defined as the complement of $\Tilde{\Lambda}$ w.r.t $\Lambda$. It can be noted that since both the terms appearing on the R.H.S. of Eq.(\ref{factc}) are nonnegative, if the L.H.S. of Eq.(\ref{factc}) is $0$ then both the individual terms on the R.H.S. must be $0$. Thus from Eq.(\ref{hb}) we must have
\begin{align}
\int_{\lambda \in \Tilde{\Lambda}} d\lambda \mu(\lambda)p_A(0|0,\lambda)p_B(1|1,\lambda)=0,\nonumber\\
\implies p_B(1|1,\lambda)=0~\forall~\lambda \in \Tilde{\Lambda}.\nonumber\\
\mbox{Since}~p_A(0|0,\lambda)>0~\&~\mu(\lambda)>0~\forall~\lambda \in \Tilde{\Lambda},\nonumber\\
\implies p_B(0|1,\lambda)=1~\forall~\lambda \in \Tilde{\Lambda}.
\end{align}
Similarly from Eq.(\ref{hc}) we must have 
\begin{align}
p_A(0|1,\lambda)=1~\forall~\lambda \in \Tilde{\Lambda}.   
\end{align}
Thus from Eq(\ref{hd}) we have
\begin{align}
&\left(\int_{\lambda \in \Tilde{\Lambda}^c}+\int_{\lambda \in \Tilde{\Lambda}}\right)d\lambda \mu(\lambda)p_A(0|1,\lambda)p_B(0|1,\lambda)=0,\nonumber\\
&\implies\int_{\lambda \in \Tilde{\Lambda}^c} d\lambda \mu(\lambda)p_A(0|1,\lambda)p_B(0|1,\lambda)\nonumber\\
&\hspace{4cm}+\int_{\lambda \in \Tilde{\Lambda}} d\lambda \mu(\lambda)=0.\label{facthd}
\end{align}
This is impossible due to the fact that $\mu(\lambda)>0~\forall~\lambda \in \Tilde{\Lambda}$, which guarantees a strictly positive contribution from the second term on the L.H.S. in Eq(\ref{facthd}). Thus a Bell local NS correlation does not satisfy all the constraints (\ref{ha}-\ref{hd}). Accordingly, a NS correlation satisfying all the constraints (\ref{ha}-\ref{hd}) must be Bell nonlocal, and  such a correlation is called Hardy's Nonlocal correlation.

\section{Analysis and implications of Theorems 3 \& 4 and Corollary 1}\label{appen-b}
In the framework of resource theory, quantum entanglement is considered to be a useful resource under the free operation of `local quantum operations and classical communication' (LOCC). In this operational paradigm, Nielson's result \cite{Nielsen1999} proves that a state $\ket{\psi}$ can be transferred to another state $\ket{\phi}$ uder LOCC {\it if and only if} $\lambda_\psi$ is majorized by $\lambda_\phi$, where $\lambda_i$'s are Schmidt coefficient of the respective states. As it turns out, for the $\mathbb{C}^2\otimes\mathbb{C}^2$ system any non-maximally entangled state $\ket{\psi}\in\mathbb{C}^2\otimes\mathbb{C}^2$ can be deterministically prepared from the maximally entangled state $\ket{\phi^+}$ using $1$-bit classical communication from Alice to Bob, whereas the reverse transformation is not possible. This implies that maximally entangled states $\ket{\phi^+}$ will surpass any non-maximally entangled state $\ket{\psi}$ in all possible tasks if classical communication is considered as a free resource. 
\begin{figure}[t!]
\centering
\includegraphics[width=0.45\textwidth]{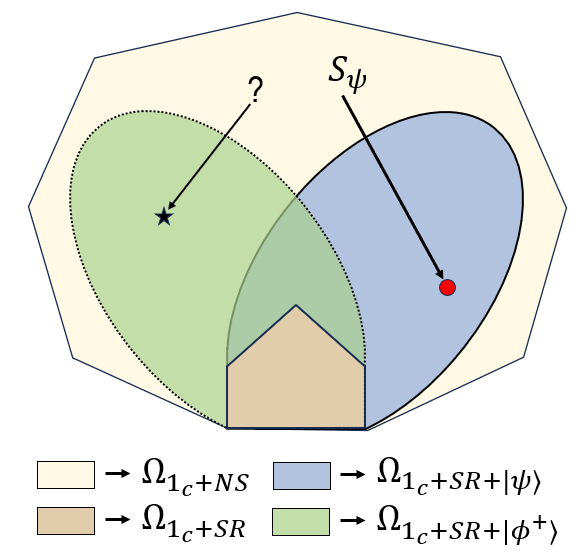}
\caption{Set of $4$-input-$4$-output channels simulable by $1$-bit of classical communication in assistance with different kind of preshared resources. $\Omega_{1_c+SR}(4,4)$ denotes the set of channels simulable by $1$-bit classical communication when unlimited amount of shared randomness is available as assistance. $\Omega_{1_c+NS}(4,4)$ corresponds to the set when communication line is assisted with arbitrary amount of NS correlation. $\Omega_{1_c+SR+\ket{\chi}}(4,4)$ denotes the set of channels simulable when arbitrary amount of SR along with the state $\ket{\chi}\in\mathbb{C}^2\otimes\mathbb{C}^2$ is available as assistance. Two convex sets are depicted for maximally and non-maximally entangled states. Red dot denotes the channel $S_\psi$ such that $S_\psi\in\Omega_{1_c+SR+\ket{\psi}}$ but $S_\psi\notin\Omega_{1_c+SR+\ket{\phi^+}}$. Star denotes a possible channel that lies within $\Omega_{1_c+SR+\ket{\phi^+}}$, but does not belong to the set $\Omega_{1_c+SR+\ket{\psi}}$.}
\label{fig3}
\end{figure}

However, the situation is quite different if classical communication is treated as a costly resource. This becomes quite evident in the reverse zero-error channel simulation task considered in our work. For instance, as discussed in our Corollary 1, there exist a $4$-input-$4$-output noisy channel $S_\psi$ of the form $S_{\mathbb{H}}$ of Eq.(4) of the main manuscript, which can be perfectly simulated by $1$-bit classical communication from Alice to Bob with the assistance of the non-maximally entangled state $\ket{\psi}$. By performing suitable local measurements on the state $\ket{\psi}$ one first obtain a $2$-input-$2$-output NS correlation that exhibits Hardy's nonlocality which according to Theorem 3 is necessary to simulate a channel of the form $S_{\mathbb{H}}$ when communication from Alice to Bob is limited to $1$-cbit. At this point it should be noted that $2$-input-$2$-output correlation obtained from the state $\ket{\phi^+}$ cannot exhibit Hardy's nonlocality \cite{Goldstein1994,Jordan1994}. However, this itself does not discard the possibility of simulating a noisy channel of the form $S_{\mathbb{H}}$ with $1$-cbit channel from Alice to Bob in assistance with $\ket{\phi^+}$. From the state $\ket{\phi^+}$ one can tries to come up with higher input-output NS correlation by performing local POVM on $\ket{\phi^+}$, which can further be used for the targeted simulation task. Our Theorem 4, however, proves that this, in-fact, is not possible. Given the state $\ket{\phi^+}$ one might aim to convert it to the state $\ket{\psi}$ following Nielson's protocol and then try to simulate the channel $S_\psi$. But, this also is not possible since the available $1$-cbit channel is consumed at the state transformation step and hence makes the simulation impossible. Therefore our results establishes a nontrivial advantage of a non-maximally entangle state over the maximally entangle state in the paradigm of limited classical communication scenario even when the non-maximally entangled state contains arbitrarily small amount of entanglement. Of course, there might be possibility of a different $4$-input-$4$-output noisy channel which can be simulated by $1$-cbit channel in assistance of $\ket{\phi^+}$, but not in assistance with $\ket{\psi}$. Although we believe that such channels should exist we could not come up with explicit such examples, and leave the question for future research. The aforesaid discussion is depicted in Fig.\ref{fig3}. 

\section{Advantage of Cabello's Nonlocal Correlation}\label{appen-c}
In the main manuscript, we have shown that Hardy's nonlocal correlation provides a communication advantage in the DMH game. Here we extend this result for nonlocal correlations exhibiting Cabello's nonlocality \cite{Cabello2002} -- which can be thought of as a generalization of Hardy's nonlocality argument. As shown by Cabello a binary input-output NS correlation $\{c(a,b|x,y)\}$ satisfying the constraints
\begin{align}
\left\{\begin{aligned}
c(0,0|0,0):=c_0>c_3,~~c(0,1|0,1):=c_5=0,\\
c(1,0|1,0):=c_{10}=0,~~c(0,0|1,1):=c_3,~~~~~\\
\mbox{with},~c(a,b|x,y):=c_{a\times2^3+b\times2^2+x\times2^1+y\times2^0}
\end{aligned}\right\}\label{Hardyc}
\end{align}
must be nonlocal in nature. To establish the communication advantage of such a correlation we consider a variant of the DMH game which will be denoted as DMH$^\prime$ and specified by the payoff matrix:
\begin{align}
\mathcal{G}_{DMH^\prime}	
\equiv
\begingroup
\setlength{\tabcolsep}{1pt} 
\renewcommand{\arraystretch}{1} 
\begin{array}{c||c|c|c|c|} 
 \mathcal{M}\textbackslash\mathcal{Z}& 1 & 2 & 3 & 4\\ \hline\hline
1& -1 & 0 & 0 & -\infty\\ \hline
2& -\infty & 0 & -1 & 0\\ \hline
3& +1 & 0 & -\infty & -\infty\\ \hline
4& -\infty & -\infty & 0 & 0\\ \hline
\end{array}
\endgroup
\label{DMH(c)-payoff}
\end{align}
Our next result limits the success probability of this game when $1$-cbit communication from Alice to Bob is allowed along with an unlimited amount of shared randomness. 
\begin{theorem}\label{theo5}
The average payoff of the DMH$^\prime$ game is upper bounded by zero while following a strategy from the set $\Omega_{1_c+SR}(4,4)$, {\it i.e.}, $\mathcal{P}(S)\le0,~\forall~S\in\Omega_{1_c+SR}(4,4)$.  
\end{theorem}
\begin{proof}
The proof follows similar reasoning as of Theorem 1 in the main manuscript. The only extreme strategies which will not trigger the bomb are
\begin{align}
\left\{\begin{aligned}
S^e_{s1}:=\begin{pmatrix}
     1 & 0 & 0 & 0\\
     0 & 0 & 1 & 0\\
     1 & 0 & 0 & 0\\
     0 & 0 & 1 & 0
     \end{pmatrix},~~
S^e_{s2}:=\begin{pmatrix}
     0 & 0 & 1 & 0\\
     0 & 0 & 1 & 0\\
     1 & 0 & 0 & 0\\
     0 & 0 & 1 & 0
     \end{pmatrix},\\
S^e_{s3}:=\begin{pmatrix}
     1 & 0 & 0 & 0\\
     0 & 0 & 0 & 1\\
     1 & 0 & 0 & 0\\
     0 & 0 & 0 & 1
     \end{pmatrix},~~
S^e_{s4}:=\begin{pmatrix}
     0 & 1 & 0 & 0\\
     0 & 1 & 0 & 0\\
     0 & 1 & 0 & 0\\
     0 & 0 & 1 & 0
     \end{pmatrix},\\
S^e_{s5}:=\begin{pmatrix}
     0 & 1 & 0 & 0\\
     0 & 0 & 1 & 0\\
     0 & 1 & 0 & 0\\
     0 & 0 & 1 & 0
     \end{pmatrix}~~
     S^e_{s6}:=\begin{pmatrix}
     0 & 0 & 1 & 0\\
     0 & 1 & 0 & 0\\
     0 & 1 & 0 & 0\\
     0 & 0 & 1 & 0
     \end{pmatrix},\\
     S^e_{s7}:=\begin{pmatrix}
     0 & 0 & 1 & 0\\
     0 & 0 & 1 & 0\\
     0 & 1 & 0 & 0\\
     0 & 0 & 1 & 0
     \end{pmatrix}~~
     S^e_{s8}:=\begin{pmatrix}
     0 & 1 & 0 & 0\\
     0 & 1 & 0 & 0\\
     0 & 1 & 0 & 0\\
     0 & 0 & 0 & 1
     \end{pmatrix},\\
     S^e_{s9}:=\begin{pmatrix}
     0 & 1 & 0 & 0\\
     0 & 0 & 0 & 1\\
     0 & 1 & 0 & 0\\
     0 & 0 & 0 & 1
     \end{pmatrix},~~~~~~~~~~~~~~~~~~
\end{aligned}\right\}
\end{align}
As all the extreme strategies result in either zero or negative payoff, the optimal payoff with $1$ bit classical communication and shared randomness is simply zero.
 \end{proof}
Next, we proceed to establish the communication advantage of nonlocal correlation exhibiting Cabello's nonlocality.
\begin{theorem}\label{theo6}
A strictly positive average payoff in the DMH$^\prime$ game is achievable with $1$-bit perfect classical channel from Alice to Bob when the channel is assisted with a $2$-input-$2$-output Cabello's nonlocal correlation.  
\end{theorem}
\begin{proof}
Once again the proof structure is the same as Theorem 2 of the main manuscript.
{\bf Alice's action:}\\
$\bullet$ Depending on $m\in\mathcal{M}$ Alice computes her input in the NS box. For $m\in\{1,3\}$ she chooses $x=0$, otherwise she choose $x=1$.\\
$\bullet$ Based on the tuple $(m,a)\in\mathcal{M}\times\mathcal{A}$ she communicates to Bob. She sends $c=0$ to Bob when $(m,a)\in\{(1,1),(2,1),(3,0),(3,1)\}$, else she sends $c=1$.\\
{\bf Bob's action:}\\
$\bullet$ The communicated bit from Alice is used as input in Bob's part of the NS box.\\
$\bullet$ Depending on the tuple $(c,b)\in C\times\mathcal{B}$ he chooses the box as follows: $(0,0)\mapsto1,~(0,1)\mapsto2,~(1,0)\mapsto3,~(1,1)\mapsto4$.

The above strategy with $1$-bit communication and $2$-input-$2$-output Cabello's correlation $\{c(ab|xy\}$ leads to the strategy matrix,
\begin{align}
S_H	
\equiv
\begingroup
\setlength{\tabcolsep}{1pt} 
\renewcommand{\arraystretch}{1} 
\begin{array}{c||c|c|c|c|} 
 \mathcal{M}\textbackslash\mathcal{Z}& 1 & 2 & 3 & 4\\ \hline\hline
1& c_8 & c_{12} & c_1 & 0\\ \hline
2& 0 & c_{14} & c_3 & c_7\\ \hline
3& c_0+c_8 & c_4+c_{12} & 0 & 0\\ \hline
4& 0 & 0 & c_3+c_{11} & c_7+c_{15}\\ \hline
\end{array}
\endgroup
\label{CabeloChannel}
\end{align}
As evident from the payoff matrix (\ref{DMH(c)-payoff}) and the strategy (\ref{CabeloChannel}), a box containing bomb will never be opened. Furthermore, assuming Charlie's choice to be completely random, we have the average payoff 
\begin{align}\label{pay1}
\mathcal{P}(S_H)=\frac{1}{4}\tr[ \mathcal{G}^{T}_{DMH^\prime}~S_H]=\frac{1}{4}(c_0-c_3)>0.
\end{align}
This completes the proof. 
\end{proof}


\begin{thebibliography}{99}
\bibitem{Self0} Although uses of finite payoffs are common in game theory, the framework is quite flexible to incorporate wide range of situations including those with $\pm\infty$ payoffs. Such extreme payoffs generally appear in financial trading games, gambling games, and decision-making games. A classic example is the ``Stag Hunt" game. If the hunters are part of a community that relies on hunting for survival then the failure to catch a stag leads to starvation and the payoff for each hunter could be considered $-\infty$.
\end{thebibliography}
\end{document}